\documentclass[a4paper,UKenglish,cleveref, autoref, thm-restate]{lipics-v2021}

\usepackage{tikz}
\usepackage{wrapfig}
\usepackage{multirow}
\usetikzlibrary{calc, intersections}
\usepackage[linesnumbered,lined,boxed,ruled]{algorithm2e}
\usepackage[normalem]{ulem}

\title{The Covering Canadian Traveller Problem Revisited} 

\newcommand{\lipsix}{Sorbonne Universit\'e, CNRS, LIP6, F-75005 Paris, France}

\newcounter{mycounter}
\newtheorem{myexample}[mycounter]{Example}

\author{Niklas Hahn}{\lipsix}{niklas.hahn@lip6.fr}{https://orcid.org/0000-0002-4929-0542}{}
\author{Michalis Xefteris}{\lipsix}{michail.xefteris@lip6.fr}{https://orcid.org/0009-0006-2894-3029}{}

\authorrunning{N. Hahn and M. Xefteris} 

\Copyright{Niklas Hahn and Michalis Xefteris} 



\category{} 

\relatedversion{} 

\acknowledgements{We would like to thank Evripidis Bampis and Bruno Escoffier for providing useful comments on the manuscript.}

\nolinenumbers 

\newcommand{\cnn}{\texttt{CNN}}
\newcommand{\newgraph}{\ensuremath{G_p^{+}}}

\begin{document}

\maketitle

\begin{abstract}
In this paper, we consider the $k$-Covering Canadian Traveller Problem ($k$-CCTP), which can be seen as a variant of the Travelling Salesperson Problem. The goal of $k$-CCTP is finding the shortest tour for a traveller to visit a set of locations in a given graph and return to the origin. Crucially, unknown to the traveller, up to $k$ edges of the graph are blocked and the traveller only discovers blocked edges online at one of their respective endpoints. The currently best known upper bound for $k$-CCTP is $O(\sqrt{k})$ which was shown in [Huang and Liao, ISAAC '12]. We improve this polynomial bound to a logarithmic one by presenting a deterministic $O(\log k)$-competitive algorithm that runs in polynomial time. Further, we demonstrate the tightness of our analysis by giving a lower bound instance for our algorithm.
\end{abstract}

\newpage

\section{Introduction}
The Canadian Traveller Problem (CTP) was introduced in 1991 by Papadimitriou and Yannakakis~\cite{PY} as an extension of the Shortest Path Problem and has applications in online route planning in road networks. The goal of the problem is to find a shortest path between a source and a destination in an unreliable graph, in which some edges may become unavailable. This can only be observed in an online manner, i.e., when reaching one of the endpoints of such an edge. More specifically, consider a connected, undirected graph $G = (V, E)$ with a source node $s \in V$, a destination node $t \in V$ and a non-negative cost function $c\colon E \rightarrow \mathbb{R^+}$ representing the cost to traverse each edge. A traveller seeks to find a path with minimum cost from $s$ to $t$. However, one or more edges might be blocked, and thus cannot be traversed. The traveller only learns that an edge is blocked when reaching one of its endpoints. When the number of blocked edges is bounded by $k$, the variant is called $k$-Canadian Traveller Problem or $k$-CTP~\cite{BNS}.

This work studies a generalization of CTP, defined in~\cite{huang}, which is called the Covering Canadian Traveller Problem (CCTP). In CCTP, one attempts to develop an efficient tour for a traveller that visits all vertices in a graph and returns to the origin (source) under the same uncertainty as that of CTP. When the number of blocked edges is bounded by $k$, the problem is called $k$-CCTP, analogous to the $k$-CTP variant of CTP.

We make two assumptions on the underlying graph model, similar to~\cite{huang}. First, we assume that the graph remains connected even if all blocked edges are removed. Second, the state of an edge, i.e., whether it is blocked or not, does not change after the traveller learns about it. The problem has practical uses in dynamic routing systems that prioritize efficient travel routes and aim to avoid traffic congestion. Since Huang and Liao introduced CCTP~\cite{huang}, there has been a notable amount of work on similar problems in the literature~\cite{follow3, follow4, follow1, follow2}. For example, Zhang et al.~\cite{zhang} studied the Steiner Travelling Salesperson Problem in which the salesperson instantly learns about new blocked edges. Shiri et al.~\cite{Shiri2} focused on how to allocate and route search-and-rescue teams to areas with trapped victims, considering the uncertainty about road conditions which may delay the operations. 

The motivation behind CCTP stems from other similar optimization problems, such as dynamic TSP and online TSP. Dynamic TSP has been studied for various different types of dynamic change, such as the addition or removal of locations, and changing pairwise distances between locations in the underlying space~\cite{Larsen, Toriello}. On the other hand, Ausiello et al.~\cite{ausiello} introduced the online TSP in which the input arrives over time, i.e., during the travel new requests (locations) appear that have to be visited by the algorithm.  The problem has many practical applications, e.g., in logistics and robotics \cite{Ascheuer,Psaraftis}. Since its introduction, a series of papers has been published on the subject~\cite{ausiello2,Bjelde,jaillet}.


As is usual in the literature on online problems, we measure the performance of our algorithm by its competitive ratio~\cite{Borodin}. This means that its performance is compared to the performance of an algorithm for the corresponding offline problem. In our setting, this would be an algorithm which knows the complete graph structure, including all blocked edges.

\paragraph*{Our Contribution}
In this paper, we focus on $k$-CCTP. Currently, the best known deterministic algorithm for $k$-CCTP is the \textsc{Cyclic Routing} algorithm by Huang and Liao~\cite{huang} with competitive ratio $O(\sqrt{k})$. We improve this bound to $O(\log k)$ by making a connection with the Online Graph Exploration problem. In the Online Graph Exploration problem, a searcher starts from a source vertex and aims to visit all vertices of an unknown but fixed graph. Upon reaching a new vertex, the server learns all incident edges and their costs. The reduction we give allows us to get a polynomial time algorithm for $k$-CCTP using an algorithm for the Exploration problem. Finally, we show that our analysis of the $O(\log k)$-competitive algorithm is tight.

\section{Related Work}
\paragraph*{Online Graph Exploration Problem}

In the Online Graph Exploration problem, defined in~\cite{pruhs}, an agent has to explore an unknown graph by starting at a given vertex, visiting all other vertices, and returning to the starting one. The agent can only move along the edges of the graph and has to pay a cost for each traversed edge.

A simple and fast algorithm that solves the problem is the
Nearest Neighbor (NN) algorithm. The algorithm selects an unexplored vertex that is cheapest to reach from the current one and visits it, repeating this process until all vertices are visited. This algorithm has been shown to have a competitive ratio of $\Theta(\log n)$ for arbitrary graphs in the Online Graph Exploration Problem~\cite{Rosenkrantz}, which is a tight bound even on planar unit-weight graphs~\cite{Fritsch, Hurkens}. Note that although the analysis in~\cite{Rosenkrantz} deals with the offline problem, the nearest neighbor can always be identified even in the online scenario and the same analysis applies. The second algorithm that achieves the $\Theta(\log n)$ bound, which is the best known upper bound for arbitrary graphs, is the hierarchical Depth First Search algorithm (\textsc{hDFS}) in~\cite{megow}. 

On the other hand, obtaining constant-competitive tours is known only for special cases of graphs, such as graphs with $k$ distinct weights, graphs with bounded genus, cycles, tadpole graphs and cactus graphs~\cite{brandt, Fritsch, pruhs, megow, miyazaki}. Conversely, the best known lower bound on the competitive ratio of an online algorithm is just $10/3$~\cite{Birx}, and despite efforts, it is still unclear whether there exists an $o(\log n)$ or even $O(1)$-competitive exploration algorithm for general graphs.

\paragraph*{Canadian Traveller Problem}
CTP has been proven to be PSPACE-complete~\cite{PY}. For the $k$-CTP variant, Bar-Noy and Schieber proposed a polynomial time algorithm that minimizes the maximum travel length~\cite{BNS}. Westphal developed a simple deterministic online algorithm for $k$-CTP that is $(2k+1)$-competitive and proved that no deterministic online algorithm can have a (strictly) better competitive ratio~\cite{WESTPHAL}. Furthermore, he showed a lower bound of $k+1$ for any randomized algorithm, even if all $s-t$ paths are node disjoint. Xu et al.~\cite{Xu} proposed a deterministic algorithm that is also $(2k+1)$-competitive for $k$-CTP and proved that a natural greedy strategy based on the available blockage information is exponential in $k$. On graphs where all $s-t$ paths are node-disjoint, a $(k+1)$-competitive randomized online algorithm is known~\cite{bender,Shiri}. Demaine et al.~\cite{Demaine} proposed a polynomial time randomized algorithm that improves the deterministic lower bound of $2k+1$ by an $o(1)$ factor for arbitrary graphs. They also showed that the competitive ratio is even better if the randomized algorithm runs in pseudo-polynomial time. Recently, Bergé et al.~\cite{berge} proved that the competitive ratio of any randomized algorithm using a specific set of strategies called memoryless cannot be better than $2k + O(1)$. Over the last few years, various other variants of CTP have been investigated~\cite{evripidis, huang2, KN}.

\paragraph*{Covering Canadian Traveller Problem}
The best known algorithm for $k$-CCTP is the one proposed in~\cite{huang} with a competitive ratio of $O(\sqrt{k})$. The algorithm, called \textsc{Cyclic Routing}, decomposes the entire route  into several rounds. In each round, the traveller attempts to visit as many vertices as possible in the graph following the visiting order (or the reverse order) of the tour  derived by Christofides' algorithm~\cite{Christofides}.

\section{Preliminaries}
In this section, we give some basic definitions. We start by giving a formal definition of the problem we study, before defining the Online Graph Exploration Problem, the performance measure and restating Christophides' algorithm for completeness. In what follows, we will denote by $G=(V,E)$ a weighted, undirected graph. We will interchangeably use the notion of ``cost'' and ``length'' for the weight of an edge. For example, a shortest tour is a tour of minimum cost.
 
\paragraph*{Definition of $\boldsymbol{k}$-CCTP}
The formal definition of CCTP is as follows. Given a complete metric graph $G = (V , E)$ with a source vertex $s \in V$, a traveller aims, beginning from $s$, to visit every other vertex in $V$ at least once and return to $s$ with as little cost as possible. However, the traveller discovers online that some edges are blocked once reaching one of their endpoints. Moreover, as mentioned earlier, two assumptions are made. First, the blocked edges cannot isolate vertices of $G$, i.e., $G$ remains connected, and second, edges remain in their state (i.e., whether they are blocked or not) forever. In this paper, we consider its variant $k$-CCTP where the number of blocked edges is bounded by $k$. 

\paragraph*{Definition of the Online Graph Exploration Problem}
The problem can be formalized as follows. Let $G=(V,E)$ be a weighted, undirected graph with $|V|=n$ vertices. The agent starts at a vertex $s \in V$ and has to visit every vertex in the graph and return to $s$. Note that the agent can visit a vertex more than once. At each step, the agent is located at a vertex $u$ and can choose to move to any of the neighboring vertices of $u$. The agent incurs a cost equal to the cost of the edge traversed. Upon arriving at a vertex $v$, the agent learns all the edges incident to $v$ and their costs. 

\paragraph*{Competitive Ratio}
A deterministic online algorithm $ALG$ for $k$-CCTP is $c$-competitive if the total cost $|ALG(\sigma)|$ accrued by $ALG$ for input $\sigma$ is at most $c \cdot |OPT(\sigma)|$. Here, $|OPT(\sigma)|$ is the total cost of an optimal tour for $\sigma$ which is computed by an offline algorithm that already knows all blocked edges.

\paragraph*{Christophides' algorithm}
We also remind the reader how Christophides' algorithm works. Christophides' algorithm on a complete metric graph $G$ can be described as follows~\cite{goodrich}:
\begin{enumerate}
        \item Create a minimum spanning tree $T$ of $G$.
        \item Find a minimum-weight perfect matching $M$ in the subgraph of $G$ that is induced by the vertices with odd degrees in $T$.
        \item Combine the edges of $M$ and $T$ to form a connected multigraph $H$.
        \item Form a Eulerian cycle in $H$.
        \item Make the circuit found in the previous step into a Hamiltonian cycle by skipping repeated vertices.
\end{enumerate}

\section{Solving \textit{k}-CCTP via Graph Exploration}

In this section, we present the results of our work. First, we show a connection between CCTP and the Online Graph Exploration problem (Theorem~\ref{reduction}). This is the crucial step to improve the upper bound of $O(\sqrt{k})$. 

The idea behind our reduction is that CCTP can be solved by an algorithm that solves the Online Graph Exploration Problem. This is possible since at every step, the traveller locally learns the real edges in both problems. The challenge here is that the algorithms for the Online Graph Exploration for arbitrary graphs have competitive ratios depending on the number of vertices $n$. 

So, how can we reduce the size of the graph in which we run an algorithm for Graph Exploration to something of size $O(k)$? First, we try to follow an approximately optimal TSP tour, skipping vertices when edges are discovered to be blocked.
Similar to the idea in \textsc{Cyclic Routing} of~\cite{huang}, we use a function \textsc{ShortCut} to achieve this.
After that, we return to the starting vertex. This way, we visit at least $n-k$ vertices of $G$. Since they do not have to be visited again, we can then use the information gathered to reduce the number of vertices in the graph on which we will run the algorithm for Graph Exploration to $O(k)$. Formally, we have the following theorem. 

\begin{theorem} \label{reduction}
If there exists an $f(k)$-competitive algorithm for the Online Graph Exploration problem on graphs with at most $k+1$ nodes and an $\alpha$-approximation algorithm for metric TSP, then there exists an $(f(k)+2\alpha)$-competitive algorithm for $k$-CCTP.
\end{theorem}

\begin{proof}
Suppose that we have \textsc{AlgoExploration}, an $f(k)$-competitive algorithm for the Online Graph Exploration problem on arbitrary graphs with at most $k+1$ nodes, and \textsc{AlgoTSP}, an $\alpha$-approximation algorithm for metric TSP. Then, we will prove that the algorithm \textsc{CompressAndExplore} that uses these algorithms as subroutines has a competitive ratio of at most $f(k)+2\alpha$ for $k$-CCTP (for pseudocode, see Algorithm~\ref{comp_expl}).

First, the algorithm runs \textsc{AlgoTSP} on the input graph $G$ to compute a TSP tour $P$. For simplicity, we relabel the vertices with respect to the tour, i.e., we assume that the tour $P$ has the order $s=v_1 \to v_2 \to \dots \to v_n \to v_1$. If an edge $\{v_i, v_j\}$ is blocked, the traveller tries to go to the next vertex in the order determined by $P$, i.e., $v_{j+1}$, or $v_1$ for $j=n$ (for pseudocode of this subroutine, see Function~\textsc{ShortCut} on page~\pageref{shortcut}). This procedure is possible since the original graph is complete. By the triangle inequality, the cost of the tour is upper bounded by the cost of tour $P$. If the traveller reaches vertex $s$, then \textsc{ShortCut} terminates. If $s$ is not reachable directly because of a blocked edge, the traveller returns to $s$ by retracing their steps. Since cost($P$) is an $\alpha$-approximation for metric TSP, we have that $\text{cost}(\textsc{ShortCut}) \le  2 \cdot \text{cost}(P) \le 2 \alpha \cdot |OPT|$, where $OPT$ is an optimal offline TSP tour on graph $G$. 

The traveller learns about all blocked edges which are adjacent to the vertices that are visited during \textsc{ShortCut}. In the procedure, all edges that are discovered to be blocked are collected in the set $E_b$. Thus, the traveller knows the whole graph (with all blocks) except for the induced (complete) subgraph formed by the unvisited vertices $U$. Let $\kappa$ be the number of vertices which remain unvisited by $ALG$ after \textsc{ShortCut}, i.e., the size of the set $U$. Then, the traveller has discovered at least $\kappa$ blocked edges, i.e., $|E_b| \ge \kappa$. 

Next, the traveller, being at $s$, has to visit the vertices in $U$. Since the true edges of the graph except for those of the induced subgraph formed by vertices in $U$ are known, it suffices to consider only the vertices in the set $U_s = U \cup \{s\}$. While the vertices in $V \setminus U_s$ themselves are not required, a shortest path between two vertices $x,y \in U_s$ might include vertices from the set $V \setminus U_s$ as intermediate nodes. This can occur when currently unknown edges between unvisited nodes are blocked. More specifically, the algorithm runs the function \textsc{Compress} (for pseudocode, see Function~\textsc{Compress} on page~\pageref{compress}).
For every pair of vertices $x,y \in U_s$, the function creates a new edge $P_{x,y}$ representing a shortest path between $x$ and $y$ such that the path consists only of edges that are known not to be blocked, i.e., edges in which at least one node has already been visited before -- if such a shortest path exists. Note that this phase of the algorithm does not incur any cost in terms of competitive ratio. Thus, the procedure creates a multigraph $G'$ which consists of vertex set $U_s$, the initial edges that connect these vertices and the ``shortest-path'' edges as described above. To better explain the steps of the algorithm we present an example of an execution of algorithm \textsc{CompressAndExplore} below (see Example~\ref{example:CAE}).

Finally, the algorithm runs \textsc{AlgoExploration} on $G'$ and visits the remaining vertices.\footnote{\textsc{AlgoExploration} solves the Exploration problem on arbitrary graphs, but $G'$ is a multigraph with at most two edges per pair of vertices. However, this does not cause a problem, since the algorithm can always select a shortest edge out of the two and the optimal solution can be computed while keeping only one edge per pair.} Every time the traveller visits a vertex, they learn all incident edges. This includes the newly added ``shortest-path'' edges, of which we know that they are feasible. If the traveller uses such a ``shortest-path'' edge, then in the final computed tour in the original graph we expand it, meaning that we use the real path that corresponds to this edge. The cost of an optimal TSP tour $OPT_{G'}$ on multigraph $G'$ is at most the cost of an optimal TSP tour $OPT_{G}(U_s)$ that only has to visit the vertex set $U_s$, i.e., the vertices that are also in $G'$, but inside the input graph $G$. To see that this holds, consider an optimal tour $OPT_G(U_s)$. Assume that it visits the vertices in $U_s$ in the order $s=x_1 \to x_2 \to \dots \to x_{|U_s|} \to x_1$. Between any vertices $x_i, x_{i+1} \in U_s$ for $i \in \{1,\dots,|U_s|-1$\} (or $x_{|U_s|}, x_1$) that are visited one after the other, $OPT_G(U_s)$ uses a shortest path. Each of these shortest paths starts in $U_s$. If it then uses an edge to another vertex in $U_s$, this edge will also be in $G'$ as each direct edge between two vertices of $U_s$ will either be blocked in both $G$ and $G'$ or not be blocked in both. Hence, we can assume that an edge $\{u,v\}$ from $u \in U_s$ to a vertex $v \notin U_s$ is taken. This is an already discovered edge, as $v$ has already been visited during \textsc{ShortCut}. Eventually, the path will re-enter into the set $U_s$ by using another already discovered edge $\{v',u'\}$ for some $v' \notin U_s$ and $u' \in U_s$. In between leaving and re-entering, all edges that were taken are also already discovered and this partial path has exactly the same length as the shortest-path edge $P_{u,u'}$ between $u,u' \in U_s$.\footnote{There might be several shortest paths with the same length, which is why the shortest path chosen for $P_{u,u'}$ and the described shortest path might differ. Still, their lengths are equal by definition.}
Continuing this argument, eventually the target vertex in $U_s$ is reached. All intermediate partial paths are inside $G'$, either since they are regular edges that also exist in $G$, or since they have been added as shortest-path edges during \textsc{ShortCut}.

The multigraph $G'$ has $\kappa+1$ vertices and at least $\kappa$ blocks have already been discovered. The number of blocked edges is at most $k$, and thus there are at most $k+1$ vertices in $G'$. So, from the hypothesis the cost incurred by \textsc{AlgoExploration} on $G'$ is at most $f(k) \cdot |OPT_{G'}|$. Since an optimal solution for visiting a subset of vertices $OPT_{G}(U_s)$ has cost at most $|OPT|$, we get the following
\begin{align*}
  \text{cost(\textsc{AlgoExploration})} \le f(k) \cdot |OPT_{G'}|
    \le f(k) \cdot |OPT_{G}(U_s)| \le f(k) \cdot |OPT| \text{ .}
\end{align*}
Overall, the algorithm has a total cost for the traveller of
\begin{align*}
  \text{cost(\textsc{CompressAndExplore})} &=  \text{cost(\textsc{ShortCut})} + \text{cost(\textsc{AlgoExploration})} \\ &\le (f(k)+2\alpha) \cdot |OPT| \text{ .} 
\end{align*}
Consequently, \textsc{CompressAndExplore} is an $(f(k)+2\alpha)$-competitive algorithm for $k$-CCTP. Note that the knowledge of $k$ does not affect the performance of the algorithm. 
\end{proof}

\begin{algorithm}[t]
\caption{\textsc{CompressAndExplore(AlgoTSP, AlgoExploration)}} \label{comp_expl}
\SetAlgoLined
\DontPrintSemicolon
\SetKwData{Left}{left}\SetKwData{This}{this}\SetKwData{Up}{up}
\SetKwFunction{Union}{Union}\SetKwFunction{FindCompress}{FindCompress}
\SetKwInOut{Input}{Input}\SetKwInOut{Output}{Output}
\SetKwInOut{Parameter}{Parameter}
\Input{A complete metric graph $G=(V,E)$ with $n$ vertices; a starting vertex $s \in V$;}
\Output{A tour that visits every vertex in $V$;}
\Parameter{\textsc{AlgoTSP}($G_1$): An algorithm that returns a TSP tour on a metric graph $G_1$; The tour has the form $s=v_1 \to v_2 \to \dots \to v_n \to v_1$;\\
    \textsc{AlgoExploration}($G_2$): An algorithm that solves the Online Graph \\Exploration problem on an arbitrary graph $G_2$ and returns a tour;}
$P \leftarrow \textsc{AlgoTSP}(G)$;\\
$G^*, U, P_1 \leftarrow \textsc{ShortCut}(G, P)$;\\
$G' \leftarrow \textsc{Compress}(G^*, U, G)$;\\
$P_2 \leftarrow \textsc{AlgoExploration}(G')$;\\
$P' \leftarrow (P_1 \to P_2)$;\\ 
\tcc{Concatenate $P_1$ and $P_2$, i.e., visit the vertices according to $P_1$, then according to $P_2$.}
\textbf{return} $P'$;
\end{algorithm}

\begin{algorithm}[t]
\label{shortcut}
\SetAlgoLined
\DontPrintSemicolon
    \SetKwFunction{FMain}{\textsc{ShortCut}}
    \SetKwProg{Fn}{Function}{:}{}
    \Fn{\FMain{$G$, $P$}}{
        \tcc{$G$ is the input graph and $P$ a TSP tour.\\
        $P$ has the form $s=v_1 \to v_2 \to \dots \to v_n \to v_1$.}
        $i  \leftarrow 1$; 
        $j  \leftarrow 2$;\\
        $U_s \leftarrow \{s\}$;
        $E_b \leftarrow \emptyset$;
        $P' \leftarrow \{s\}$;\\
        \tcc{Path $P'$ which the traveller follows is built.} 

        \While{$j \le n$}{
        Add all newly discovered blocked edges $\{v_i,x\}$, with $x \in V \setminus \{v_i\}$, to $E_b$;\\
        {\eIf{$\{v_i,v_j\}$ is not blocked}
            {$P' \leftarrow (P' \to v_j) $; \tcp*{Append $v_j$ to $P'$}    
            $ i \leftarrow j;$}
            {$ U_s \leftarrow U_s \cup \{v_j\}$;
            }
        }
        $j \leftarrow j+1$;
        }
        \eIf{$\{v_i,v_1\}$ is blocked}
            {Return to $s$ following $P'$ backwards;\\
            $P' \leftarrow$ Concatenate the path $P'$ and the reverse of $P'$ to return to $s$; 
            }
            {$P' \leftarrow (P' \to  v_1)$;}
        $G^* \leftarrow (V, E \setminus E_b)$;\\
        \textbf{return} $G^*,  U_s, P'$;
}
\textbf{end Function}
\end{algorithm}

\begin{algorithm}[t]
\label{compress}
\SetAlgoLined
\DontPrintSemicolon
    \SetKwFunction{FMain}{\textsc{Compress}}
    \SetKwProg{Fn}{Function}{:}{}
    \Fn{\FMain{$G^*$, $U_s$, $G$}}{
        \tcc{$G^*$ is the graph without the discovered blocked edges and $U_s$ the set of the remaining unvisited vertices in the graph together with the starting vertex $s$.}
        $E' \leftarrow \{\{x,y\} \in E \mid x,y \in U_s\}$; \\
        \tcc{$E'$ is the subset of edges with unknown state, i.e., of $\{x,y\}$ with $x,y \in U_s$.}
        $G' \leftarrow (U_s, E')$;\\
        $H \leftarrow (V, E \setminus E')$; \\
        \tcc{$H$ includes all edges with a known state, since in every edge at least one vertex has already been visited.}
        Let $U_s = \{ v_1',v_2', \dots, v_{|U_s|}'\}$;\\
        \For{$i\gets1$ \KwTo $|U_s|$}{
            \For{$j\gets i+1$ \KwTo $|U_s|$}{
                Find a shortest path $P_{i,j}$ from $v_i'$ to $v_j'$ in $H$; \\
                $c_{i,j} \leftarrow$ total cost of $P_{i,j}$;\\
                Add an edge $\{v_i',v_j'\}$ with cost $c_{i,j}$ to $G'$;                
        }
        }
        \textbf{return} $G'$;
}
\textbf{end Function}
\end{algorithm}

\begin{remark*}
In the proof, we allow $k \ge n-1$ as long as the resulting graph remains connected. The analysis of the competitive ratio of $O(\sqrt{k})$ in~\cite{huang} requires $k < n-1$.
\end{remark*}

\begin{myexample}\label{example:CAE}
Fig.~\ref{example} shows an example of \textsc{CompressAndExplore}. The traveller begins at vertex $s = v_1$ and moves in a counterclockwise direction. The given TSP tour by \textsc{AlgoTSP} here is $v_1 \to v_2 \to \dots \to v_{16} \to v_1$. The solid lines represent the tour that the traveller follows during \textsc{ShortCut} due to the discovered blocked edges (red dashed lines). The traveller follows the shortcut path $v_1 \to v_2 \to v_4 \to v_5 \to v_9 \to v_{10} \to v_{11} \to v_{14} \to v_{16}$ and after visiting vertex $v_{16}$, they return back to $s$ following the same path backwards. Next, the algorithm runs \textsc{Compress} and gets $G'$. Multigraph $G'$ contains $s$, the remaining unvisited vertices and at most two edges between each pair of these vertices. Between $v_i$ and $v_j$ there is the edge $\{v_i,v_j\}$ (which may be blocked) and possibly the ``shortest-path'' edge $P_{i,j}$. The cost of $P_{i,j}$ is the cost of the shortest path from $v_i$ to $v_j$ in which each edge has at least one endpoint outside of $G'$. In the example, a possible case for $i=1$ and $j=3$ is shown on the right with $P_{1,3}$ being the path $v_1 \to v_4 \to v_3$. 

Finally, the algorithm runs \textsc{AlgoExploration} on $G'$. The traveller visits all remaining vertices, returns to $s$ and the algorithm terminates.

\begin{figure}
    \caption{An example of algorithm \textsc{CompressAndExplore}.} \label{example}
     \centering
     \begin{tikzpicture}[scale=0.85]
  \def \n {16}
  \def \radius {3cm}
  \def \circradius {3pt}
  \def \s {1}
  \coordinate (P\s) at ({360/\n * (\s - 1)}:\radius);
  \draw[] (P\s) circle (\circradius);
  \draw (\s*360/16: 3.55cm) node{$v_{\s} = s$};
  \foreach \s in {2,...,\n}
  {
    \coordinate (P\s) at ({360/\n * (\s - 1)}:\radius);
    \draw[] (P\s) circle (\circradius);
    \draw (\s*360/16: 3.4cm) node{$v_{\s}$};
  }

  \foreach \s in {1,2,3,5,6,10,11,12,15}
  {
    \draw[fill=black] (P\s) circle (\circradius);  
  }

    \path[name path=line] ($(P2)$) -- ($(P3)$);
    \path[name path=circle1] (P2) circle (\circradius);
    \path[name path=circle2] (P3) circle (\circradius);
    \draw[name intersections={of=line and circle1, by=intersection1}, name intersections={of=line and circle2, by=intersection2}] (intersection1) -- (intersection2);

    \path[name path=line] ($(P3)$) -- ($(P5)$);
    \path[name path=circle1] (P3) circle (\circradius);
    \path[name path=circle2] (P5) circle (\circradius);
    \draw[name intersections={of=line and circle1, by=intersection1}, name intersections={of=line and circle2, by=intersection2}] (intersection1) -- (intersection2);

    \path[name path=line] ($(P3)$) -- ($(P4)$);
    \path[name path=circle1] (P3) circle (\circradius);
    \path[name path=circle2] (P4) circle (\circradius);
    \draw[dashed, red, name intersections={of=line and circle1, by=intersection1}, name intersections={of=line and circle2, by=intersection2}] (intersection1) -- (intersection2);

    \path[name path=line] ($(P5)$) -- ($(P6)$);
    \path[name path=circle1] (P5) circle (\circradius);
    \path[name path=circle2] (P6) circle (\circradius);
    \draw[name intersections={of=line and circle1, by=intersection1}, name intersections={of=line and circle2, by=intersection2}] (intersection1) -- (intersection2);

    \path[name path=line] ($(P6)$) -- ($(P7)$);
    \path[name path=circle1] (P6) circle (\circradius);
    \path[name path=circle2] (P7) circle (\circradius);
    \draw[dashed, red, name intersections={of=line and circle1, by=intersection1}, name intersections={of=line and circle2, by=intersection2}] (intersection1) -- (intersection2);

    \path[name path=line] ($(P6)$) -- ($(P8)$);
    \path[name path=circle1] (P6) circle (\circradius);
    \path[name path=circle2] (P8) circle (\circradius);
    \draw[dashed, red, name intersections={of=line and circle1, by=intersection1}, name intersections={of=line and circle2, by=intersection2}] (intersection1) -- (intersection2);

    \path[name path=line] ($(P6)$) -- ($(P9)$);
    \path[name path=circle1] (P6) circle (\circradius);
    \path[name path=circle2] (P9) circle (\circradius);
    \draw[dashed, red, name intersections={of=line and circle1, by=intersection1}, name intersections={of=line and circle2, by=intersection2}] (intersection1) -- (intersection2);

    \path[name path=line] ($(P6)$) -- ($(P10)$);
    \path[name path=circle1] (P6) circle (\circradius);
    \path[name path=circle2] (P10) circle (\circradius);
    \draw[name intersections={of=line and circle1, by=intersection1}, name intersections={of=line and circle2, by=intersection2}] (intersection1) -- (intersection2);

    \path[name path=line] ($(P10)$) -- ($(P11)$);
    \path[name path=circle1] (P10) circle (\circradius);
    \path[name path=circle2] (P11) circle (\circradius);
    \draw[name intersections={of=line and circle1, by=intersection1}, name intersections={of=line and circle2, by=intersection2}] (intersection1) -- (intersection2);

    \path[name path=line] ($(P11)$) -- ($(P12)$);
    \path[name path=circle1] (P11) circle (\circradius);
    \path[name path=circle2] (P12) circle (\circradius);
    \draw[name intersections={of=line and circle1, by=intersection1}, name intersections={of=line and circle2, by=intersection2}] (intersection1) -- (intersection2);

    \path[name path=line] ($(P12)$) -- ($(P13)$);
    \path[name path=circle1] (P12) circle (\circradius);
    \path[name path=circle2] (P13) circle (\circradius);
    \draw[dashed, red, name intersections={of=line and circle1, by=intersection1}, name intersections={of=line and circle2, by=intersection2}] (intersection1) -- (intersection2);

    \path[name path=line] ($(P12)$) -- ($(P14)$);
    \path[name path=circle1] (P12) circle (\circradius);
    \path[name path=circle2] (P14) circle (\circradius);
    \draw[dashed, red, name intersections={of=line and circle1, by=intersection1}, name intersections={of=line and circle2, by=intersection2}] (intersection1) -- (intersection2);

    \path[name path=line] ($(P12)$) -- ($(P15)$);
    \path[name path=circle1] (P12) circle (\circradius);
    \path[name path=circle2] (P15) circle (\circradius);
    \draw[name intersections={of=line and circle1, by=intersection1}, name intersections={of=line and circle2, by=intersection2}] (intersection1) -- (intersection2);

    \path[name path=line] ($(P15)$) -- ($(P16)$);
    \path[name path=circle1] (P15) circle (\circradius);
    \path[name path=circle2] (P16) circle (\circradius);
    \draw[red, dashed, name intersections={of=line and circle1, by=intersection1}, name intersections={of=line and circle2, by=intersection2}] (intersection1) -- (intersection2);

    \path[name path=line] ($(P15)$) -- ($(P1)$);
    \path[name path=circle1] (P15) circle (\circradius);
    \path[name path=circle2] (P1) circle (\circradius);
    \draw[name intersections={of=line and circle1, by=intersection1}, name intersections={of=line and circle2, by=intersection2}] (intersection1) -- (intersection2);

    \path[name path=line] ($(P1)$) -- ($(P2)$);
    \path[name path=circle1] (P1) circle (\circradius);
    \path[name path=circle2] (P2) circle (\circradius);
    \draw[red,dashed,name intersections={of=line and circle1, by=intersection1}, name intersections={of=line and circle2, by=intersection2}] (intersection1) -- (intersection2);

    \path[name path=line] ($(P10)$) -- ($(P1)$);
    \path[name path=circle1] (P10) circle (\circradius);
    \path[name path=circle2] (P1) circle (\circradius);
    \draw[red,dashed,name intersections={of=line and circle1, by=intersection1}, name intersections={of=line and circle2, by=intersection2}] (intersection1) -- (intersection2);

    \path[name path=line] ($(P3)$) -- ($(P15)$);
    \path[name path=circle1] (P3) circle (\circradius);
    \path[name path=circle2] (P15) circle (\circradius);
    \draw[red,dashed,name intersections={of=line and circle1, by=intersection1}, name intersections={of=line and circle2, by=intersection2}] (intersection1) -- (intersection2);
\end{tikzpicture}
\hfill
\begin{tikzpicture}[scale=0.85]
  \def \n {16}
  \def \radius {3cm}
  \def \circradius {3pt}

  \draw (0,0) node{$G'$};

  \def \s {1}
  \coordinate (P\s) at ({360/\n * (\s - 1)}:\radius);
  \draw[dotted] (P\s) circle (\circradius);
  \draw (\s*360/16: 3.55cm) node{$v_{\s} = s$};

  \foreach \s in {2,...,\n}
  {
    \coordinate (P\s) at ({360/\n * (\s - 1)}:\radius);
    \draw[dotted] (P\s) circle (\circradius);
    \draw[] (\s*360/16: 3.4cm) node{$v_{\s}$};
  }
  \coordinate (P3') at ({360/\n * (3 - 1)}:\radius-0.34cm);
  \coordinate (P15') at ({360/\n * (15 - 1)}:\radius-0.4cm);
  \coordinate (P1') at ({360/\n * (1 - 1)}:\radius-0.3cm);

  \foreach \s in {4,7,8,9,13,14,16}
  {
    \draw[] (P\s) circle (\circradius);  
  }
  \draw[fill=black] (P2) circle (\circradius);
  \draw[black] (3*360/16: 2.25cm) node{$P_{1,3}$};
\path[name path=line] ($(P2)$) -- ($(P4)$);
\path[name path=circle1] (P2) circle (\circradius);
\path[name path=circle2] (P4) circle (\circradius);
\draw[black,name intersections={of=line and circle1, by=intersection1}, name intersections={of=line and circle2, by=intersection2}] (intersection1) -- (intersection2);

\path[name path=line] ($(P2)$) -- ($(P5)$);
\path[name path=circle1] (P2) circle (\circradius);
\path[name path=circle2] (P5) circle (\circradius);
\draw[black,name intersections={of=line and circle1, by=intersection1}, name intersections={of=line and circle2, by=intersection2}] (intersection1) -- (intersection2);

\path[name path=line] ($(P5)$) -- ($(P4)$);
\path[name path=circle1] (P4) circle (\circradius);
\path[name path=circle2] (P5) circle (\circradius);
\draw[black,name intersections={of=line and circle1, by=intersection1}, name intersections={of=line and circle2, by=intersection2}] (intersection1) -- (intersection2);
\definecolor{blue-green}{rgb}{0.0, 0.87, 0.87}
\definecolor{bondiblue}{rgb}{0.0, 0.58, 0.71}
\draw[bondiblue, fill=blue-green, fill opacity=0.20, dashed] plot [smooth cycle] 
coordinates {(P1') (P2) (P3') (P4) (P7) (P8) (P9) (P13) (P14) (P15') (P16)};
\end{tikzpicture}
\end{figure}
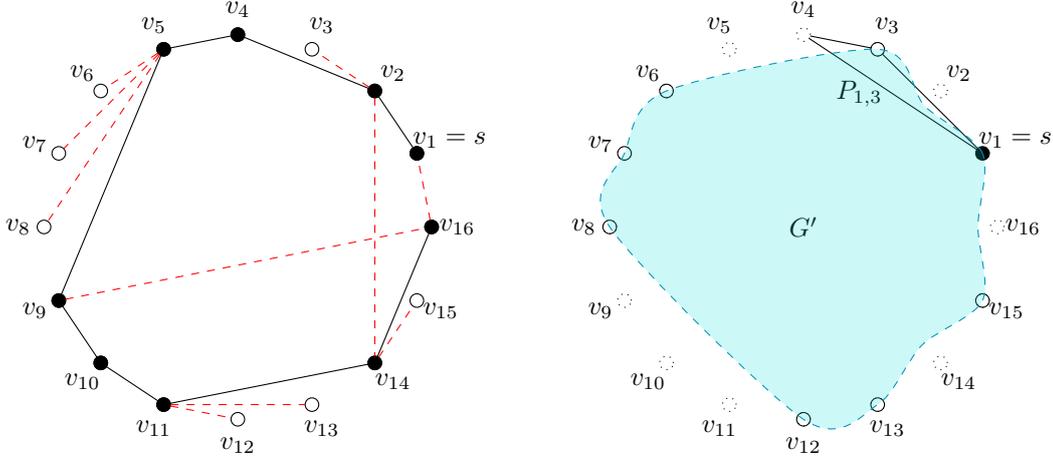
\end{myexample}

Now, we can use \textsc{CompressAndExplore} with Christophides' algorithm for metric TSP and the Nearest Neighbor for the Online Graph Exploration problem. Since this algorithm uses Christophides' algorithm and then Nearest Neighbor, we refer to it by $\cnn$.

\begin{corollary}
$\cnn$ has a competitive ratio of $O(\log k)$ for $k$-CCTP.
\end{corollary}
\begin{proof}
    Christophides' algorithm gives a $3/2$-approximation for metric TSP. On the other hand, NN yields a competitive ratio of $O(\log n)$ for the Graph Exploration problem in an arbitrary graph, where $n$ is the number of vertices in the graph. Thus, from Theorem~\ref{reduction} we get that $\cnn$ has a competitive ratio of $3+O(\log(k+1)) = O(\log k)$.
\end{proof}

To demonstrate that the above analysis is tight, the following theorem presents a family of instances that achieves a competitive ratio of $\Omega(\log k)$ and therefore proves the analysis to be tight.

\begin{theorem}
    There exists a family of instances for which $\cnn$ has a competitive ratio of the $\Omega(\log k)$.
\end{theorem}
\begin{proof}
    We will use the graph presented in~\cite{Hurkens} to lower bound the competitive ratio of the Nearest Neighbor algorithm. For an integer $p \ge 1$, the graph $G_p = (V_p, E_p)$ consists of a chain of $2^p-1$ triangles. $G_p$ has $2^p$ vertices in its lower level, and $2^p-1$ vertices in its upper level. The left-most vertex in the lower level is denoted by $l_p$, the right-most by $r_p$ and the central vertex in the upper level is denoted by $m_p$. All edges in $G_p$ have an equal cost of $1$.
    
    For our instance, we slightly modify the graph by adding another vertex $u$ to the left of $l_p$ with an edge $\{u,l_p\}$ of cost $1$. We also set $s=l_p$ as the starting vertex. Since the input for $k$-CCTP is a complete graph, we also need to add some more edges. All edges from $u$ to the other vertices have a cost of $1$, and all other new edges have a cost of $2$. All these edges will be blocked edges. We call this new graph $\newgraph$. 
    
    The resulting graph has $k=\Theta(n^2)$ blocked edges and clearly satisfies the triangle inequality. We illustrate the non-blocked part of $\newgraph$ for $p=3$ in Fig.~\ref{fig:tightness}.
    \begin{figure}[h]
\caption{Graph $\newgraph$ (for $p=3$). $\newgraph$ is used to show tightness of the $O(\log k)$-competitive ratio.} \label{fig:tightness}
\centering
\begin{tikzpicture}[scale=0.92]
  \begin{scope}[xshift=-1.8cm]
    \coordinate (u) at (210:1) {};
    \node[shift={(-0,-0.35)}] at (u) {\footnotesize $u$};
    \draw[fill=black, draw=black] (u) circle (2pt);
  \end{scope}

  \coordinate (A1) at (90:1) {};
  \coordinate (B1) at (210:1) {};
  \coordinate (C1) at (-30:1) {};

  \draw[fill=black, draw=black] (A1) circle (2pt);
   \node[shift={(-0,-0.35)}] at (B1) {\footnotesize $l_3$};
  \draw[fill=black, draw=black] (B1) circle (2pt);
  \draw[fill=black, draw=black] (C1) circle (2pt);

  \draw (u) -- (B1);
  \draw (A1) -- (B1) -- (C1) -- (A1);
  
  \begin{scope}[xshift=1.8cm]
    \coordinate (A2) at (90:1) {};
    \coordinate (B2) at (-30:1) {};
    \draw[fill=black, draw=black] (A2) circle (2pt);
    \draw[fill=black, draw=black] (B2) circle (2pt);
    \draw (A2) --  (B2) -- (C1) -- (A2);
  \end{scope}

  \begin{scope}[xshift=3.6cm]
    \coordinate (A3) at (90:1) {};
    \coordinate (B3) at (-30:1) {};
    \draw[fill=black, draw=black] (A3) circle (2pt);
    \draw[fill=black, draw=black] (B3) circle (2pt);
    \draw (A3) --  (B3) -- (B2) -- (A3);
  \end{scope}

  \begin{scope}[xshift=5.4cm]
    \coordinate (A4) at (90:1) {};
    \coordinate (B4) at (-30:1) {};
    \draw[fill=black, draw=black] (A4) circle (2pt);
    \node[shift={(0,0.35)}] at (A4) {\footnotesize $m_3$};
    \coordinate (B4) at (-30:1) {};
    \draw[fill=black, draw=black] (B4) circle (2pt);
    \draw[fill=black, draw=black] (A4) circle (2pt);
    
    \draw (A4) --  (B4) -- (B3) -- (A4);

  \end{scope}

  \begin{scope}[xshift=7.2cm]
    \coordinate (A5) at (90:1) {};
    \coordinate (B5) at (-30:1) {};
    \draw[fill=black, draw=black] (A5) circle (2pt);
    \draw[fill=black, draw=black] (B5) circle (2pt);
    \draw (A5) --  (B5) -- (B4) -- (A5);
  \end{scope}

  \begin{scope}[xshift=9cm]
    \coordinate (A6) at (90:1) {};
    \coordinate (B6) at (-30:1) {};
    \draw[fill=black, draw=black] (A6) circle (2pt);
    \draw[fill=black, draw=black] (B6) circle (2pt);
    \draw (A6) --  (B6) -- (B5) -- (A6);
  \end{scope}

  \begin{scope}[xshift=10.8cm]
    \coordinate (A7) at (90:1) {};
    \coordinate (B7) at (-30:1) {};
    \draw[fill=black, draw=black] (A7) circle (2pt);
    \draw[fill=black, draw=black] (B7) circle (2pt);
     \node[shift={(0,-0.35)}] at (B7) {\footnotesize $r_3$};
    \draw (A7) --  (B7) -- (B6) -- (A7);
  \end{scope}
\end{tikzpicture}

\end{figure}
    
    In the first step of Christophides' algorithm, a minimal spanning tree is constructed. One possible MST $T$ is a path from $u$ to $r_p$. The nodes with uneven degree in $T$ are the nodes $u$ and $r_p$, so for the matching, the edge between $u$ and $r_p$ is added. This results in a simple cycle of all nodes. The MST and the matching edge are illustrated in Fig.~\ref{fig:MST}.

    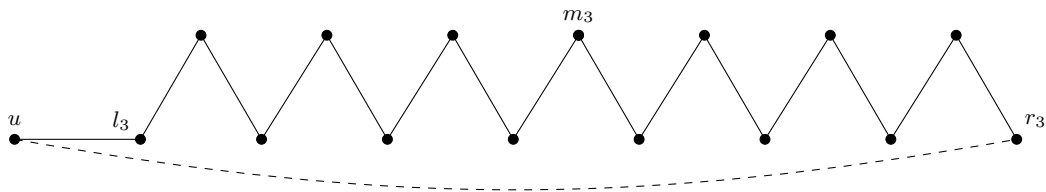
\begin{figure}[h]
\caption{The MST and the matching edge (dashed line) in $G_3^+$.} \label{fig:MST}
\centering
\begin{tikzpicture}[scale=0.92]
  \begin{scope}[xshift=-1.8cm]
    \coordinate (u) at (210:1) {};
    \node[shift={(-0,0.25)}] at (u) {\footnotesize $u$};
    \draw[fill=black, draw=black] (u) circle (2pt);
  \end{scope}

  \coordinate (A1) at (90:1) {};
  \coordinate (B1) at (210:1) {};
  \coordinate (C1) at (-30:1) {};

  \draw[fill=black, draw=black] (A1) circle (2pt);
   \node[shift={(-0.25,0.25)}] at (B1) {\footnotesize $l_3$};
  \draw[fill=black, draw=black] (B1) circle (2pt);
  \draw[fill=black, draw=black] (C1) circle (2pt);

  \draw (u) -- (B1);
  \draw (A1) -- (B1);
  \draw (C1) -- (A1);
  
  \begin{scope}[xshift=1.8cm]
    \coordinate (A2) at (90:1) {};
    \coordinate (B2) at (-30:1) {};
    \draw[fill=black, draw=black] (A2) circle (2pt);
    \draw[fill=black, draw=black] (B2) circle (2pt);
    \draw (A2) --  (B2);
    \draw (C1) -- (A2);
  \end{scope}

  \begin{scope}[xshift=3.6cm]
    \coordinate (A3) at (90:1) {};
    \coordinate (B3) at (-30:1) {};
    \draw[fill=black, draw=black] (A3) circle (2pt);
    \draw[fill=black, draw=black] (B3) circle (2pt);
    \draw (A3) --  (B3);
    \draw (B2) -- (A3);
  \end{scope}

  \begin{scope}[xshift=5.4cm]
    \coordinate (A4) at (90:1) {};
    \coordinate (B4) at (-30:1) {};
    \draw[fill=black, draw=black] (A4) circle (2pt);
    \node[shift={(0,0.25)}] at (A4) {\footnotesize $m_3$};
    \coordinate (B4) at (-30:1) {};
    \draw[fill=black, draw=black] (B4) circle (2pt);
    \draw[fill=black, draw=black] (A4) circle (2pt);
    
    \draw (A4) --  (B4);
    \draw (B3) -- (A4);

  \end{scope}

  \begin{scope}[xshift=7.2cm]
    \coordinate (A5) at (90:1) {};
    \coordinate (B5) at (-30:1) {};
    \draw[fill=black, draw=black] (A5) circle (2pt);
    \draw[fill=black, draw=black] (B5) circle (2pt);
    \draw (A5) --  (B5);
    \draw (B4) -- (A5);
  \end{scope}

  \begin{scope}[xshift=9cm]
    \coordinate (A6) at (90:1) {};
    \coordinate (B6) at (-30:1) {};
    \draw[fill=black, draw=black] (A6) circle (2pt);
    \draw[fill=black, draw=black] (B6) circle (2pt);
    \draw (A6) --  (B6);
    \draw (B5) -- (A6);
  \end{scope}

  \begin{scope}[xshift=10.8cm]
    \coordinate (A7) at (90:1) {};
    \coordinate (B7) at (-30:1) {};
    \draw[fill=black, draw=black] (A7) circle (2pt);
    \draw[fill=black, draw=black] (B7) circle (2pt);
     \node[shift={(0.25,0.25)}] at (B7) {\footnotesize $r_3$};
    \draw (A7) --  (B7);
    \draw (B6) -- (A7);
  \end{scope}
  \draw[dashed] (u) edge[bend right = 10] (B7);
\end{tikzpicture}

\end{figure}
    
    The TSP-tour can then be chosen to be $s=l_3 \to u \to r_3 \to \dots \to l_3$. This means that in \textsc{ShortCut}, only node $u$ would be visited besides $l_p$ as the direct edges from $u$ to any other node (besides $l_p$) are blocked. At the end of \textsc{ShortCut}, the traveller returns to $l_p$. 
    
    After \textsc{ShortCut}, the remaining graph would thus be the original graph $G_p$ from~\cite{Hurkens}. 
    We use the following lemma to prove that there exists a TSP-tour in $G_p$ which starts (and ends) in $l_p$ which is found by NN that has a length of $(p+4)\cdot 2^{p-1}-2$. We will prove the lemma below.
    
\begin{lemma}[Based on {\cite[Lemma 1]{Hurkens}}]\label{lem:hurkens}
    There exists a NN-based TSP tour on $G_p$ which starts in $l_p$ and visits $m_p$ as final vertex before returning back to $l_p$. The tour has length $(p+4)\cdot 2^{p-1}-2$.
\end{lemma}

    Using this result, the total cost of the described TSP-tour is $(p+4)2^{p-1}$, whereas an optimal TSP-tour has cost $2+3(2^{p}-1)$, namely visiting $u$ and optimally visiting $G_p$ by going in a zig-zag motion from left to right (as shown in the MST in Fig.~\ref{fig:MST}) and returning using the lower edges of the triangle, thereby using each edge of the triangles exactly once. This gives us a ratio of \[
    \frac{(p+4)\cdot 2^{p-1}}{2+3\cdot (2^{p}-1)} = \frac{(p+4) \cdot 2^{p-1}}{6\cdot 2^{p-1}-1} \ge \frac{p+4}{6} = \Omega(p) = \Omega(\log n)\enspace.
    \]
\end{proof}

\begin{proof}[Proof of Lemma~\ref{lem:hurkens}]
    We split the tour into two parts. In the first part, all vertices are visited, and in the second part, the traveller returns to $l_p$. 
    
    The second part has length $1+2^{p-1}-1 = 2^{p-1}$. This is true because the whole graph has been discovered and the traveller can take the shortest path from $m_p$ to $l_p$, which is going down to the left point of the middle triangle and then traversing the $2^{p-1}-1$ many triangles on the left side to reach $l_p$.
    
    Hence, to show the Lemma, it remains to show that there exists a NN-route to visit all vertices in $G_p$ which has length $(p+3)\cdot 2^{p-1}-2 = (p+4)\cdot2^{p-1} - 2 - 2^{p-1}$, starting at $l_p$ and ending at $m_p$. We prove this by induction. For $p=1$, $G_p$ consists of a single triangle, and the route $l_1 \to r_1 \to m_1$ has length $2 = (1+3) \cdot 2^{1-1} - 2$. 

    For the inductive step, we observe that $G_p$ can be constructed from two copies of $G_{p-1}$ and an additional vertex $m_p$ (and three additional edges). Let $G_{p-1}^{(l)}$ be the left copy and $G_{p-1}^{(r)}$ be the right copy. Then, the new edges are $\{r_{p-1}^{(l)}, l_{p-1}^{(r)}\}, \{r_{p-1}^{(l)}, m_p\}$ and $\{l_{p-1}^{(r)},m_p\}$. This is illustrated in Fig.~\ref{fig:G_pRecursion}. By the induction hypothesis, there exists an NN-route in $G_{p-1}^{(l)}$ starting in $l_{p} = l_{p-1}^{(l)}$ and ending in $m_{p-1}^{(l)}$ with length $(p-1+3) \cdot 2^{p-1-1}-2$. The two nearest unvisited neighbors to $m_{p-1}^{(l)}$ are $m_p$ and $l_{p-1}^{(r)}$ with equal distance $2^{p-2}+1$. By going to $l_{p-1}^{(r)}$, the sub-route from $l_{p-1}^{(r)}$ to $m_{p-1}^{(r)}$ of length $(p+2) \cdot 2^{p-2} - 2$ can then be found by NN. Note that throughout this route, $m_p$ will never be closer to the current vertex than any other unvisited vertex in the current sub-route and thus will not be visited before $m_{p-1}^{(r)}$. Finally, $m_p$ needs to be visited, which requires an additional cost of $2^{p-2}+1$. Overall, there exists an NN-route from $l_p$ to $m_p$ with length \[
    2\cdot ((p+2) \cdot 2^{p-2} -2) + 2\cdot(2^{p-2}+1) = (p+2) \cdot 2^{p-1} - 4 + 2^{p-1} + 2 = (p+3) \cdot 2^{p-1} - 2 \enspace.\] This concludes the proof.

\end{proof}
        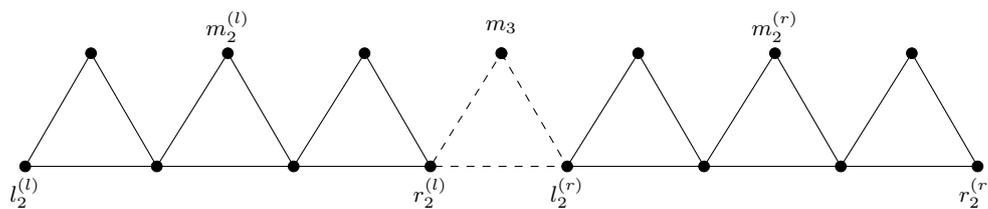
\begin{figure}[h]
\caption{Graph $G_p$ (for $p=3$) constructed from two copies of $G_{p-1}$, with the joining edges denoted by the dashed lines.} \label{fig:G_pRecursion}
\centering
    \begin{tikzpicture}
        
  \coordinate (A1) at (90:1) {};
  \coordinate (B1) at (210:1) {};
  \coordinate (C1) at (-30:1) {};

  \draw[fill=black, draw=black] (A1) circle (2pt);
   \node[shift={(-0,-0.35)}] at (B1) {\footnotesize $l_2^{(l)}$};
  \draw[fill=black, draw=black] (B1) circle (2pt);
  \draw[fill=black, draw=black] (C1) circle (2pt);

  \draw (A1) -- (B1) -- (C1) -- (A1);
  
  \begin{scope}[xshift=1.8cm]
    \coordinate (A2) at (90:1) {};
    \coordinate (B2) at (-30:1) {};
    \draw[fill=black, draw=black] (A2) circle (2pt);
    \node[shift={(0,0.35)}] at (A2) {\footnotesize $m_2^{(l)}$};
    \draw[fill=black, draw=black] (B2) circle (2pt);
    \draw (A2) --  (B2) -- (C1) -- (A2);
  \end{scope}

  \begin{scope}[xshift=3.6cm]
    \coordinate (A3) at (90:1) {};
    \coordinate (B3) at (-30:1) {};
    \draw[fill=black, draw=black] (A3) circle (2pt);
    \node[shift={(0,-0.35)}] at (B3) {\footnotesize $r_2^{(l)}$};
    \draw[fill=black, draw=black] (B3) circle (2pt);
    \draw (A3) --  (B3) -- (B2) -- (A3);
  \end{scope}

  \begin{scope}[xshift=5.4cm]
    \coordinate (A4) at (90:1) {};
    \coordinate (B4) at (-30:1) {};
    \draw[fill=black, draw=black] (A4) circle (2pt);
    \node[shift={(0,0.35)}] at (A4) {\footnotesize $m_3$};
    \coordinate (B4) at (-30:1) {};
    \draw[fill=black, draw=black] (B4) circle (2pt);
    \node[shift={(0,-0.35)}] at (B4) {\footnotesize $l_2^{(r)}$};
    \draw[fill=black, draw=black] (A4) circle (2pt);
    
    \draw[dashed] (A4) --  (B4) -- (B3) -- (A4);

  \end{scope}

  \begin{scope}[xshift=7.2cm]
    \coordinate (A5) at (90:1) {};
    \coordinate (B5) at (-30:1) {};
    \draw[fill=black, draw=black] (A5) circle (2pt);
    \draw[fill=black, draw=black] (B5) circle (2pt);
    \draw (A5) --  (B5) -- (B4) -- (A5);
  \end{scope}

  \begin{scope}[xshift=9cm]
    \coordinate (A6) at (90:1) {};
    \coordinate (B6) at (-30:1) {};
    \draw[fill=black, draw=black] (A6) circle (2pt);
    \node[shift={(0,0.35)}] at (A6) {\footnotesize $m_2^{(r)}$};
    \draw[fill=black, draw=black] (B6) circle (2pt);
    \draw (A6) --  (B6) -- (B5) -- (A6);
  \end{scope}

  \begin{scope}[xshift=10.8cm]
    \coordinate (A7) at (90:1) {};
    \coordinate (B7) at (-30:1) {};
    \draw[fill=black, draw=black] (A7) circle (2pt);
    \draw[fill=black, draw=black] (B7) circle (2pt);
     \node[shift={(0,-0.35)}] at (B7) {\footnotesize $r_2^{(r)}$};
    \draw (A7) --  (B7) -- (B6) -- (A7);
  \end{scope}
\end{tikzpicture}

\end{figure}

Finally, we remark that $\cnn$ takes polynomial time. The procedures \textsc{ShortCut} and \textsc{Compress} run in polynomial time as the required shortest paths can be computed in polynomial time. Since Christophides' algorithm and Nearest Neighbor also have polynomial time complexity, so does $\cnn$.

\section{Concluding Remarks}
In this work, we considered the Covering Canadian Traveller Problem with up to $k$ blocked edges. We improved the upper bound to $O(\log k)$ by drawing an interesting connection to the Online Graph Exploration problem. Further, we showed the tightness of our analysis. 

Our reduction implies immediate consequences of future work on the respective other problem. For one, it allows an improvement of the lower bound on the Graph Exploration problem using a general lower bound on $k$-CCTP. Currently, the best known bound for the Graph Exploration problem is $10/3$. Tightening this gap would be a very interesting result. 
Second, an improved algorithm for the Graph Exploration problem immediately gives rise to a better algorithm and upper bound on $k$-CCTP. 

Nevertheless, already an improved algorithm for $k$-CCTP or a lower bound on the Graph Exploration problem would be of independent interest without exploiting our reduction and thus provides another challenging direction of future research.

\bibliography{bibliography}

\end{document}